\documentclass[journal,onecolumn,draftclsnofoot,12pt]{IEEEtran}

\usepackage{graphicx}
\usepackage{amsfonts}
\usepackage{chngcntr}
\usepackage{times}
\usepackage{latexsym}
\usepackage{amssymb}
\usepackage{amsmath}
\usepackage{cite}
\usepackage{verbatim}
\usepackage{subfigure}
\usepackage{multirow}
\usepackage{lastpage}
\usepackage{calc}
\usepackage{eso-pic}
\usepackage{tabularx}
\usepackage{ifthen}
\usepackage{epstopdf}

\newtheorem{theorem}{Theorem}

\newcounter{as}
\newcounter{le}
\newcounter{cor}[theorem]

\newcounter{def}

\newtheorem{assumption}[as]{Assumption}

\newtheorem{corollary}[cor]{Corollary}
\counterwithin{cor}{theorem}

\newtheorem{definition}[def]{Definition}

\newtheorem{lemma}[le]{Lemma}

\newtheorem{remark}[theorem]{Remark}
\newenvironment{proof}{ \textbf{Proof:} }{ \hfill $\Box$}

\newcommand{\figref}[1]{{Fig.}~\ref{#1}}

\def\bb0{{\mathbb{0}}}

\def\bb{{\mathbf{b}}}

\def\bx{{\mathbf{x}}}
\def\by{{\mathbf{y}}}

\def\b0{{\mathbf{0}}}

\def\bC{{\mathbf{C}}}

\def\bO{{\mathbf{O}}}

\def\bX{{\mathbf{X}}}
\def\bY{{\mathbf{Y}}}

\def\bbE{{\mathbb{E}}}

\def\bbP{{\mathbb{P}}}

\def\bbR{{\mathbb{R}}}

\def\cA{\mathcal{A}}

\def\sf0{{\mathsf{0}}}

\usepackage{footnote}
\allowdisplaybreaks
\begin{document}

\title{Analysis of Blockage Effects on \\Urban Cellular Networks}

\author{
Tianyang Bai, Rahul Vaze, and Robert W. Heath, Jr.\\

\thanks{T. Bai and R. W. Heath Jr. are with The University of Texas at Austin, Austin, TX, USA.
R. Vaze is with Tata Institute of Fundamental Research, Mumbai, India (email: tybai@utexas.edu, vaze@tcs.tifr.res.in, rheath@utexas.edu). This material is based upon work supported by the National Science Foundation under Grant No. 1218338 and 1319556.

The material in this paper was presented in part at the International Conference on Signal Processing and Communications (SPCOM) \cite{bai2012}. }}

\maketitle
\begin{abstract}
Large-scale blockages like buildings affect the performance of urban cellular networks, especially at higher frequencies. Unfortunately, such blockage effects are either neglected or characterized by oversimplified models in the analysis of cellular networks. Leveraging concepts from random shape theory, this paper proposes a mathematical framework to model random blockages and analyze their impact on cellular network performance. Random buildings are modeled as a process of rectangles with random sizes and orientations whose centers form a Poisson point process on the plane. The distribution of the number of blockages in a link is proven to be Poisson random variable with parameter dependent on the length of the link. A path loss model that incorporates the blockage effects is proposed, which matches experimental trends observed in prior work. The model is applied to analyze the performance of cellular networks in urban areas with the presence of buildings, in terms of connectivity, coverage probability, and average rate. Analytic results show while buildings may block the desired signal, they may still have a positive impact on network performance since they can block significantly more interference.
\end{abstract}


\section{Introduction}
Penetration losses due to densely located buildings make it hard to predict coverage of cellular networks in urban areas. Such blockage effects become more severe in systems of higher frequencies, and may limit performance of emerging millimeter-wave cellular networks \cite{Pi2011,Anderson2004}. Blockages impact system functions like user association and mobility. For instance, a mobile user may associate with a further base station with line-of-sight transmission rather than a closer base station that is blocked. Traditionally, blockage effects are incorporated into the shadowing model, along with reflections, scattering, and diffraction \cite{Goldsmith2005}. Shadowing of different links is often modeled by using a log-normal distributed random variable, with the variance determined from measurements. Unfortunately, this approach does not capture the distance-dependence of blockage effects: intuitively speaking, the longer the link, the more buildings are likely to intersect it, hence more shadowing is likely to be experienced.

Cellular networks are becoming less regular as a variety of demand-based low power nodes are being deployed \cite{Ghosh2012}. Moreover, as urban areas are built out, even the locations of the macro and micro base stations are becoming random and less like points in a hexagonal grid \cite{Taylor2012}. Mathematical tools like stochastic geometry make it possible to analyze cellular networks with randomly located infrastructure \cite{Brown2000,Andrews2011,Dhillon2012,Heath2012,Haenggi2009}. With stochastic geometry, the locations of the infrastructure are often assumed to be distributed according to a spatial point process, usually a homogeneous Poisson point process (PPP) for tractability. In \cite{Brown2000}, a shotgun cellular system in which base stations are distributed as a PPP was shown to lower bound a hexagonal system in terms of certain performance metrics. In\cite{Andrews2011}, analytical expressions were proposed to compute average performance metrics, such as the coverage probability and the average rate, of a cellular network with PPP distributed base stations. The results of \cite{Andrews2011} were extended to multi-tier networks in \cite{Dhillon2012}. In \cite{Heath2012}, a hybrid model in which only interferers are modeled as a PPP outside a fixed-size cell was proposed to characterize the site-specific performance of cells with different sizes, rather than the aggregate performance metrics of the entire system. The stochastic geometry model was also applied to investigate the topic of network connectivity \cite{Haenggi2009}. One limitation of prior work in \cite{Brown2000, Andrews2011,Dhillon2012,Heath2012,Haenggi2009} is that penetration losses due to buildings in urban areas were not explicitly incorporated. Blockage effects were either neglected for simplicity \cite{Andrews2011,Dhillon2012,Haenggi2009} or incorporated into a log-normal shadowing random variable \cite{Brown2000,Heath2012}.

There are two popular approaches to incorporate blockages on wireless propagation. One method is using ray tracing to perform site-specific simulations \cite{Rizk1997,Schaubach1992}. Ray tracing requires accurate terrain information, such as the location and size of the blockages in the network to generate the received signal strength given a base station deployment. Ray tracing trades the complexity of numerical computation for an accurate site-specific solution. The second approach is to establish a stochastic model to characterize the statistics of blockages that provides acceptable estimation of the blockage effects with only a few parameters. An advantage of stochastic models is that it may be possible to analyze general networks. In \cite{Franceschetti1999,Marano1999,Marano2005}, urban areas were modeled as random lattices, which are made up of sites of unit squares. Each site is occupied by a blockage with some probability. It was assumed that signals reflect when impinging blockages with no power losses. Methods from percolation theory were applied to derive a closed-form expression of the propagation depth into the random lattice. The study did not take penetration losses into account. Moreover, its application is better suited for regular Manhattan-type cities where all blockages on the plane have a similar size and orientation. In \cite{Franceschetti2004}, the lattice model from \cite{Franceschetti1999} was extended by removing some of the restrictions. A constant power loss of signal was assumed when it impinged at the blockages. Orientations of blockages were also randomly distributed. Such refinements were achieved by representing blockages as nuclei from a point process, where the difference of the sizes of the blockages was not explicitly incorporated. In addition, a common limitation found in prior work is that the building height is ignored by restricting the model to the plane of $\bbR^2$. The models in \cite{Franceschetti1999,Marano1999,Marano2005} are not compatible with stochastic geometry network model that provides a convenient analysis of cellular networks \cite{Andrews2011}.

The main contribution of this paper is to propose a mathematical framework to model blockages with random sizes, locations, and orientations in cellular networks using concepts from random shape theory. We remove the restrictions on the orientations and sizes of the blockages as in prior lattice models. Specifically, random buildings in urban area are modeled as a process of rectangles with random sizes and orientations whose centers form a PPP, which is more general than the line segment process we used previously in \cite{bai2012}. We also extend the blockage model to incorporate the height of the transmitter, receiver, and buildings. The proposed framework extends to more general cases where blockages need not be rectangles but may be any convex shape. The proposed blockage model is also compatible with the PPP cellular network model to analyze network-level performance.

Based on the proposed model, we derive the distribution of power loss due to blockages in a link and apply it to analyze the performance of cellular networks with impenetrable blockages. Analytical results indicate that while buildings complicate the propagation environment by blocking line-of-sight links, they improve system performance in covered area by blocking more interference.

The main limitation of our work is that we only consider the direct propagation path, and ignore signal reflections, which are an component of wireless transmission. We defer the incorporation of reflections as a topic of future work.
We also assume that the blockages experienced by a link are independent, which neglects correlation induced by large buildings. Simulation results show that the error of such approximation is minor and acceptable. Compared with our prior work in \cite{bai2012}, in this paper we provide a general mathematical framework to model random blockages, and evaluate the network performance with blockages more precisely, which is based on the distribution rather than the moments of interference.

The paper is organized as follows. In the next section, after reviewing some concepts from random shape theory, we describe the system model where blockages are modeled as a random process of rectangles. In Section \ref{sec:quantif}, we derive the distribution of blockage number on a link, and apply it to quantify penetration losses of a link. We also extend results to incorporate height of the blockage. In Section \ref{apply}, we analyze blockage effects on a specific case of cellular networks in which blockages are impenetrable to signals, in terms of coverage probability, average achievable rate, and network connectivity. Simulation results and comparison with the prior lattice model are provided in Section \ref{sec:simu}. Conclusions are drawn in Section \ref{conclu}.

\textbf{Notation}: We use the following notation throughout this paper: bold lower-case letters $\bx$ are used to denote vector, bold upper-case letters $\bX$ are used to denote points (locations) in $n$-dimensional Euclidean space $\bbR^n$, non-bold letters $x,X$ used to denote scalar values, and caligraphic letters $\cA$ to denote sets. Using this notation, $\bX\bY$ is the segment connecting $\bX$ and $\bY$, $|\bX\bY|$ is the length of $\bX\bY$, and $V(\mathcal{A})$ is the volume of the set $\mathcal{A}\in\bbR^n$. We denote the origin of $\bbR^n$ as $\bO$. We use $\bbE$ to denote expectation, and $\bbP$ to denote probability.
\section{System Model} \label{sec:math}
In this section we described the blockage model for urban cellular networks. We first review concepts from random shape theory~\cite{Cowan1989,Stoyan1995}, which are used in our model formulation. Random shape theory is a branch of advanced geometry that formalizes random objects in space \cite{Cowan1989,Stoyan1995}. Then we describe our system model in which  randomly located buildings are represented by a process of random rectangles.


\subsection{Background on Random Shape Theory}

In this section, we will explain the concepts for $\bbR^n$. In this paper, however, we focus on the case where objects lie in the $\bbR^2$ plane. We will regard the height of a blockage as an independent mark associated with the object.

Let ${\cal S}$ be a set of objects which are closed and bounded (with finite area and perimeter) in $\bbR^n$. For example, ${\cal S}$ could be a collection of balls with different volumes in $\bbR^3$, or a combination of line segments, rectangles, ellipses on $\bbR^2$ space. For each element $s\in {\cal S}$, a point is defined to be its {\it center}. The center does not necessarily have to be the geographic center of the object: any well defined point will suffice. When an object is not symmetric in the space, it is also necessary to define the orientation of the object. The orientation is characterized by a directional unit vector originating from its center.

The concept of random object process (ROP) will be used throughout the paper. A random object process is constructed in the following way. First randomly sample objects from ${\cal S}$ and place the centers of these objects in $\bbR^n$ at points generated by a point process ${\cal P}$. Second, determine the orientation of each object according to some distribution. Consider a random object process of line segments on $\bbR^2$ as an example. First the segments are sampled from $\cal S$. Through this process, the length of the segments are chosen from some distribution and their midpoints are placed according to some point process. Then the orientation is determined as the angle between the directional vector of each line segment and the $x$-axis according to some distribution ranging on $[0,2\pi)$.

In general, an ROP is complicated, especially when correlations exist between objects or between the sampling, location, orientation of an object. In this paper, we focus on a special class of object processes known as a Boolean scheme, which satisfies the following properties.
\begin{itemize}
\item The center points form a PPP.
\item For all objects $s\in {\cal S}$, the attributes of an object, e.g. orientation, shape, and volume, are mutually independent.
\item For a specific object, its sampling, location, and orientation are also independently.
\end{itemize}

Note that the PPP property of the centers guarantees that the locations of different objects are independent. Hence the attributes of interests, such as the size, location, and orientation of the objects, are chosen independently in Boolean schemes. Such assumptions of independence provide tractability in the analysis of network models.

The Minkowski sum (also known as dilation) in Euclidean space helps us to extend results obtained in the special case of rectangle blockages to the general cases of convex objects. 
\begin{definition}
The Minkowski sum of two compact sets $\mathcal{A}$ and $\mathcal{B}$ in $\bbR^n$ is
\begin{eqnarray}
\mathcal{A}\oplus \mathcal{B}=\cup_{\bx\in\mathcal{A},\by\in\mathcal{B}}\:(\bx+\by).
\end{eqnarray}
\end{definition}
Note that for any compact set $\mathcal{A}$ and $\mathcal{B}$ in $\bbR^n$, the volume of their Minkowski sum $V(\mathcal{A}\oplus \mathcal{B})$ is finite.

\subsection{Cellular Network Model with Random Buildings} \label{sec:shadow}
We consider a random cellular network with a single tier of base stations, whose locations are determined from a Poisson point process. We use a Boolean scheme of random rectangles to model randomly located buildings. The key assumptions made in our model are summarized as follows.

\begin{assumption}[PPP Base Station]\label{assump:PPP}
The base station locations form a homogeneous PPP $\{\bX_i\}$ on the $\bbR^2$ plane with density $\mu$. A fixed transmission power $P_\mathrm{t}$ is assumed for each base station. A typical user, located at the origin, will be used for performance analysis. Denote the link from base station $\bX_i$ to the typical user as $\bO\bX_i$, and $\left|\bO\bX_i\right|=R_i$.
\end{assumption}
\begin{assumption}[Boolean Scheme Blockages]\label{assum:boolean}
Blockages are assumed to form a Boolean scheme of rectangles. The centers of the rectangles $\{\mathbf{C}_k\}$ form a homogeneous PPP of density $\lambda$. The lengths $L_k$ and widths $W_k$ of the rectangles are assumed to be i.i.d. distributed according to some probability density function $f_L(x)$ and $f_W(x)$ respectively. The orientation of the rectangles $\Theta_k$ is assumed to be uniformly distributed in $(0, 2\pi]$.
\end{assumption}

Note that by Assumption \ref{assum:boolean} and the definition of a Boolean scheme, for any fixed index $k$, it follows that $L_k$, $W_k$, and $\Theta_k$ are independent. The Boolean scheme of blockages is completely characterized by the quadruple $\left\{\mathbf{C}_k,L_k,W_k,\Theta_k\right\}$ as the object set of the buildings is defined by $\left\{L_k,W_k,\Theta_k\right\}$. Moreover, we define a location $\mathbf{Y}\in\mathbb{R}^2$ is {\it indoor} or {\it contained by a blockage} if there exists a blockage $D_k\in\left\{\mathbf{C}_k,L_k,W_k,\Theta_k\right\}$, such that $\mathbf{Y}\in D_k$.

\begin{assumption}[Independent Height]
Each blockage is marked with a height that is i.i.d. given by $H_k$ for the $k$-th blockage. Let the probability density function of $H_k$ be $f_H(x)$.
\end{assumption}

In Section \ref{sec:K}, in the first step of the analysis, we will ignore height by restricting our model to $\bbR^2$ space. Extensions to incorporate height are provided in Section \ref{height}.
\begin{assumption}[No Reflections]
The propagation mechanism of reflection is neglected.
\end{assumption}

Discussion on the reflection paths due to blockages can be found in \cite{Marano1999,Marano2005}. While a limitation, we defer detailed treatment of reflections to future work, as it greatly complicates the application of stochastic geometry for performance analysis.
\begin{assumption}[Rayleigh fading, No Noise]
Each link experiences i.i.d. small-scale Rayleigh fading and the network is interference limited, i.e. thermal noise is neglected.
\end{assumption}

It is shown in Section \ref{sec:simu} that ignoring noise power causes minor loss of accuracy in a network with dense base stations.

Let $P_i$ be the power received by the mobile user from base station $\bX_i$, $g_i$ be the small-scale fading term of link $\bO\bX_i$, which is an exponential random variable of mean 1, $\alpha$ be the exponent of path loss, which is normally between 2 and 6, and $K_i$ the number of buildings on the link $\bO\bX_i$. Define $\gamma_{ik}$ as the ratio of penetration power losses caused by the $k$-th $(0<k\le K_i)$ blockage on $\bO\bX_i$, which takes value in $[0,1]$, for the blockages attenuate signal power when ignoring reflections. Based on our assumptions,
\begin{eqnarray}
P_{i}=\frac{M g_i\:\prod_{k=0}^{K_i}\gamma_{ik}}{R_i^{\alpha}},\label{powerloss}
\end{eqnarray}
where $M$ is a constant determined by the signal frequency, antenna gains, and the transmitted power $P_\mathrm{t}$, and the reference path loss distance. Note that $\prod_{k=0}^{K_i}\gamma_{ik}$ is the penetration loss of power caused by blockages, which we aim to quantify in this paper.

In one special case, if all buildings in the networks have the same ratio of power losses, i.e $\forall i,k,\:\gamma_{ik}=\gamma$, then (\ref{powerloss}) can be simplified as
\begin{eqnarray}
P_{i}=\frac{M g_i\:\gamma^{K_i}}{R_i^{\alpha}},\label{powerloss2}
\end{eqnarray}
where $\gamma$ can be evaluated by the average power loss ratio through a building in the area.
If all buildings are impenetrable, then $\forall i,k,\:\gamma_{ik}=0$. {\footnote{ To simplify notation, in this paper we define $0^0=1$. In the impenetrable case, when $K_i=0$, it follows that $\gamma^{K_i}=0^0=1$.}} The {\it impenetrable} case is a good approximation for networks with large buildings, where there are many walls inside, or for millimeter-wave networks, where signals suffer from more severe penetration losses through solid materials \cite{Anderson2004}. Applying (\ref{powerloss2}), the signal-to-interference ratio when the mobile user is served by base station $\bX_i$ can be written as
\begin{eqnarray}
\mbox{SIR}(\bO\bX_{i})&=&\frac{g_{i}\gamma^{K_i}R_{i}^{-\alpha}}{\sum_{\ell:\ell\neq i}g_{\ell}\gamma^{K_\ell}R_{\ell}^{-\alpha}}.\label{eqn:SIR}
\end{eqnarray}
We will use (\ref{eqn:SIR}) to examine blockage effects on network performance in Section \ref{apply}.


\section{Quantification of Blockage Effects}\label{sec:quantif}
In this section, we quantify the power losses due to blockages on a given link $\bO\bX_i$ in the cellular network. Based on the system model, we derive the distribution of $S_i=\prod_{k=0}^{K_i}\gamma_{ik}$, the ratio of power losses due to blockages on $\bO\bX_i$. Towards that end, the distribution of the blockage number $K_i$ is investigated in $\bbR^2$, and thereafter generalized to incorporate height. Then we introduce a systematic method to calculate the distribution of $S_i$ in general cases using the Laplace transform. As the distribution of $S_i$ is difficult to obtain in closed form, we approximate it using a beta distribution by matching the moments. Last we analyze the effect of blockage on the link budget of a single link. We show that blockages on average introduce an additional exponential decay term into the path loss formula, which also matches the results from field experiments in prior work \cite{Marano2005}. Since we focus on the blockage effects on a single link $\bO\bX_i$, the index for link $i$, is omitted in this section.
\subsection{Distribution of the Number of Blockages}\label{sec:K}
The distribution of the number of blockages is required to quantify the effect of penetration losses.  In this section, we show that K is a Poisson distributed random variable.

We denote the point process that is formed by centers of the rectangles with lengths in $(\ell, \ell+\mathrm{d}\ell)$, widths in $(w, w+\mathrm{d}w)$, and orientations in $(\theta,  \theta+\mathrm{d}\theta)$ as $\Phi(\ell,w,\theta)$. Note that $\Phi(\ell,w,\theta)$ is a subset (partition) of the center point process $\{\bC_k\}$, and is a PPP according to the following lemma.
\begin{lemma}\label{lem:1} $\Phi(\ell, w, \theta)$ is a PPP with the density of $\lambda_{\ell,w,\theta}=\lambda f_L(\ell)\mathrm{d}\ell f_W(w)\mathrm{d}w f_\Theta(\theta)\mathrm d \theta$. If $(\ell_1,w_1,\theta_1)\ne(\ell_2,w_2,\theta_2)$, then $\Phi(\ell_1,w_1,\theta_1)$ and $\Phi(\ell_2,w_2,\theta_2)$ are independent processes.
\end{lemma}
\begin{proof}
From the definition of a Boolean scheme, $\{\bC_k\}$ is a PPP, and $\{L_k\}$, $\{W_k\}$, and $\{\Theta_k\}$ are sequences of i.i.d. random variables. Since the Poisson law is preserved by independent thinning \cite{Baccelli2009}, $\Phi(\ell, w, \theta)$ is a PPP. Moreover, if $(\ell_1,w_1,\theta_1)\ne(\ell_2,w_2,\theta_2)$, $\Phi(\ell_1,w_1,\theta_1)$ and $\Phi(\ell_2,w_2,\theta_2)$ are disjoint sets of points, and therefore independent processes.
\end{proof}

Define a collection of blockages as
\begin{eqnarray*}
\mathcal{B}(\ell,w, \theta)=\{(\bC_k,L_k,W_k,\Theta_k),\bC_k\in\Phi(\ell,w, \theta)\},
\end{eqnarray*}
which consists of blockages with lengths $(\ell, \ell+\mathrm{d}\ell)$, widths $(w, w+\mathrm{d}w)$, and orientations $(\theta,  \theta+\mathrm{d}\theta)$. Let $J(\ell,w,\theta)$ be the number of blockages, which belong to the subset $\mathcal{B}(\ell,w,\theta)$ and cross the link $\bO\bX$.
\begin{lemma}\label{lem:2}
$J(\ell,w,\theta)$ is a Poisson random variable with mean $\bbE[J(\ell,w,\theta)]=\lambda_{\ell,w,\theta}(R\ell|\sin\theta|+R w|\cos\theta|+\ell w)$, where $R$ is the length of the link $\bO\bX$.
\end{lemma}
\proof
As shown in Fig. \ref{Miniskisum}, a rectangle from $\mathcal{B}(\ell,w,\theta)$ intersects the link $\bO\bX$ if and only if its center falls in the region $PQSTUV$. Hence $J(\ell,w,\theta)$ equals the number of points of $\Phi(\ell,w,\theta)$ falling in the region $PQSTUV$. Let the area of region $PQSTUV$ be $S(\ell,w,\theta)$, then
\begin{eqnarray*}
S(\ell,w,\theta)&=&R\:|\sin(\phi+\theta)|\sqrt{w^2+\ell^2}+w\ell\\
&=&R\left(|\sin(\theta)|\frac{\ell}{\sqrt{w^2+\ell^2}}+|\cos(\theta)|\frac{w}{\sqrt{w^2+\ell^2}}\right)\sqrt{w^2+\ell^2}+w\ell\\
&=&R\ell|\sin\theta|+Rw|\cos\theta|+\ell w.
\end{eqnarray*}
By Lemma \ref{lem:1}, $\Phi(\ell,w,\theta)$ is a PPP of density $\lambda_{\ell,w,\theta}$. The number of points of $\Phi(\ell,w,\theta)$ falling in the region $PQSTUV$ is a Poisson variable with mean $\lambda_{\ell,w,\theta}\:S(\ell,w,\theta)$. Consequently, $J(\ell,w,\theta)$ is a Poisson variable with mean
\begin{align*}
\bbE[J(\ell,w,\theta)]&=\lambda_{\ell,w,\theta}\:S(\ell,w,\theta)\\
&=\lambda_{\ell,w,\theta}\times (R\ell|\sin\theta|+Rw|\cos\theta|+\ell w). \tag*{\endproof}
\end{align*}
\begin{figure} [!ht]
\centerline{
\includegraphics[width=0.5\columnwidth]{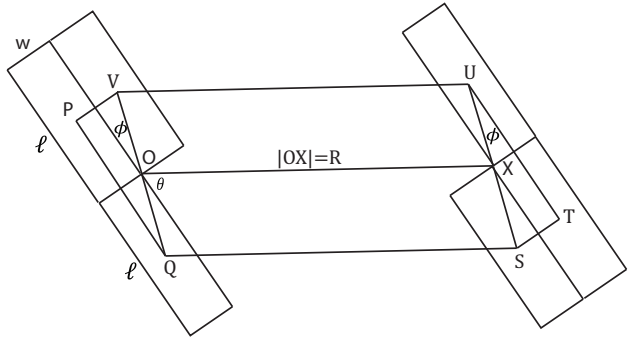}
   }
\caption{$OX$ is the link of distance $R$. $P$, $Q$, $S$, $T$, $U$, and $V$ are the centers of the corresponding rectangles. A rectangle of $B(\ell,w,\theta)$ intersects $OX$ if only if its center falls in the region $PQSTUV$, which is made up of parallelogram $QSUV$ and right triangles $PQV$ and $TSU$.}
\label{Miniskisum}
\end{figure}
Next, recall that $K$ is the total number of blockages crossing the link $\bO\bX$. We calculate the distribution of $K$ in the following theorem.
\begin{theorem} \label{thm:distK}
$K$ is a Poisson random variable with the mean $\beta R+p$, where $\beta=\frac{2\lambda(\bbE[W]+\bbE[L])}{\pi}$, and $p=\lambda\bbE[L]\bbE[W]$.
\end{theorem}
\proof
By Lemma \ref{lem:1} and Lemma \ref{lem:2}, $J(\ell,w,\theta)$ are independent Poisson random variables for different values of the tuple $(\ell,w,\theta)$. Note that for any realization of the blockage distribution, $K=\sum_{\ell,w,\theta}J(\ell,w,\theta)$ always holds. Since superpositions of independent Poisson random variables are still Poisson, $K$ is Poisson distributed. Its expectation can be computed as
\begin{align*}
\bbE[K]&=\sum_{\ell,w,\theta} K(\ell,w,\theta)\\
&=\int_{L}\int_{W}\int_{\Theta}\lambda\left(R\ell|\sin\theta|+R w|\cos\theta|+\ell w\right)\:f_L(\ell)\mathrm{d}\ell\: f_W(w)\mathrm{d}w\: \frac{1}{2\pi}\mathrm d \theta\\
&=\frac{2\lambda (\bbE[L]+\bbE[W])}{\pi}R+\lambda\bbE[L]\bbE[W]\\
&=\beta R+p. \tag*{\mbox{\endproof}}
\end{align*}
According to Theorem \ref{thm:distK}, the average number of blockages on a link is proportional to the length of the link, which matches the intuition that the longer the link is, the more blockages are likely to appear on that link. Also, when $W\equiv0$, i.e using line segments instead of rectangles to describe blockages, $\bbE[K]=\frac{2\lambda\bbE[L]R}{\pi}$, which matches our previous results in \cite{bai2012}.


The probability of line-of-sight propagation through a link can be evaluated through the following corollary. Our results on line-of-sight probability match the models used in 3GPP standards \cite{3GPPTR36.8142010}, in which the line-of-sight probability also decays exponentially as the distance increases, and typical values of $\beta$ in different scenarios are selected via measurements. 
\begin{corollary}\label{cor:free}
The probability that a link of length $R$ admits line-of-sight propagation, i.e no blockages cross the link, is
$\bbP(K=0)=\mathrm{e}^{-(\beta R+p)}$.
\end{corollary}
We can also evaluate the probability that a user is located inside a building in the following lemma.
\begin{corollary}\label{cor:pointcover}
The probability that a location in $\mathbb{R}^2$ is contained by a blockage is
$1-\mathrm{e}^{-p}\approx p$. 
\end{corollary}
\begin{proof}
This follows immediately from Corollary \ref{cor:free} by taking $R=0$. The approximation is based on the fact that $\lim_{p\to0}\frac{1-\mathrm{e}^{-p}}{p}=1.$ 
\end{proof}
Therefore, $p$ is an approximation of the fraction of the land covered by blockages in the investigated area. The approximation error is caused by neglecting the overlapping of blockages, which does exist in the Boolean scheme model. If buildings are not allowed to overlap, $p=\lambda\bbE[L]\bbE[W]$ becomes the exact evaluation of blockage coverage. Intuitively speaking, with small $p$, which indicates the blockages are sparsely distributed, blockages are unlikely to overlap. Therefore the error due to overlap is negligible. This provides a way to estimate the parameters of the model
based on the actual distribution of blockages. For example, the value of $p$ in an area can be roughly estimated using Google maps by dividing the sum area of all the buildings divided by the total area. $\bbE[L]$ and $\bbE[W]$ can be evaluated by the average size of the buildings in the area. Given the value of $p$, $\bbE[L]$, and $\bbE[W]$, $\beta$ can be evaluated by Theorem \ref{thm:distK}. 
\begin{remark}\label{remk:1}
Note the fact that the area $S(\ell,w,\theta)=V\left(\bO\bX\oplus D_{\ell,w,\theta}\right)$, where $D_{\ell,w,\theta}$ is an element of $\mathcal{B}(\ell,w,\theta)$, i.e a rectangle of length $\ell$, width $w$, and orientation $\theta$. Hence Lemma \ref{lem:2} can be rewritten as
$\bbE[J(\ell,w,\theta)]=\lambda_{\ell,w,\theta}V\left(OX\oplus D_{\ell,w,\theta}\right)$. Similarly, the result in Theorem \ref{thm:distK} can be also rewritten using Minkowski addition as
$\bbE[K]=\lambda\bbE\left[V\left(\bO\bX \oplus D\right)\right],$
where $D\in\{L_k,W_k,\Theta_k\}$ is a typical element from the object set of the rectangle Boolean scheme.
\end{remark}

Inspired by Remark \ref{remk:1}, we use the Minkowski sum to extend the results for rectangle blockages in Theorem \ref{thm:distK} to more general cases, where blockages can have any compact and convex shapes.
\begin{theorem}\label{thm:general}
Let $\mathcal{S}$ be the object set of the blockages, which can be made up of any compact and convex sets in $\bbR^2$. The number of blockages crossing a link $\bO\bX$ of length $R$ is still a Poisson random variable with mean $\mathbb{E}[K]=\lambda \bbE[V(\bO\bX\oplus D)]$, where $D\in\mathcal{S}$.
\end{theorem}
\begin{proof}
The proof is similar to that of Theorem \ref{thm:distK}, and is omitted.
\end{proof}

\subsection{Incorporating Building Height}\label{height}
In this section we extend the system model to incorporate height. For a link $\bO\bX$ of length $R$ in $\bbR^2$, $H_{\mathrm{B}}$ is the height of the base station, $H_{\mathrm{U}}$ is the height of the mobile user. Without loss of generality, assume $H_\mathrm{B}>H_\mathrm{U}$. The height of $k$-th blockage $H_k$ is i.i.d. according to some probability density function $f_H(x)$, and independent of $\{\bC_k,L_k,\Theta_k\}$. For simplicity, we assume $W\equiv0$ in this part, i.e. use line segments process to describe random buildings. The case of rectangle process can also be  extended to incorporate building height in a similar way, however, the expressions will be more complicated.

Denote $\hat{K}$ as the number of blockages that effectively block the direct propagation of the link $\bO\bX$ when considering height of blockages. Note that even if the projection of a building on the ground crosses $\bO\bX$, in practice it might not be tall enough to actually block the link as indicated in Fig. \ref{fig:height}.

\begin{figure} [!ht]
\centerline{
\includegraphics[width=0.5\columnwidth]{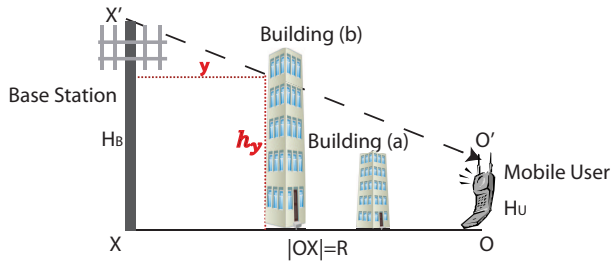}
   }
\caption{The transmitter locating at $\bX$ has a height of $H_t$, while the mobile receiver has a height of $H_r$. Not all buildings which cross $\bO\bX$ blockage the actual propagation path $\bO'\bX'$ in $\bbR^3$, such as building $(a)$ in the figure. If a building intersecting $\bO\bX$ at a point $y$ away from the transmitter $\bX$ effectively blocks $\bO'\bX'$ if and only if its height is larger than $h_y$ as building $(b)$ in the figure. }
\label{fig:height}
\end{figure}
\begin{theorem}
Considering height, the number of effective blockages of a link of length $R$, denoted by $\hat{K}$, is a Poisson random variable with $\bbE[\hat{K}]=\eta\bbE[K]$, where $\bbE[K]=\frac{2\lambda\bbE[L]R}{\pi}$, and
$$\eta=1-\int_{0}^{1}\int_{0}^{sH_\mathrm{B}+(1-s)H_\mathrm{U}}f_H(h)\mathrm{d}h\mathrm{d}s.$$
$\eta$ can be interpreted as the probability that given a building crossing $\bO\bX$, it also blocks $\bO\rq \bX\rq$ as shown in Fig. \ref{fig:height}.
\end{theorem}
\begin{proof}
Consider a building intersecting the link $\bO\bX$ at the point which is at a horizontal distance $y$ away from $\bX$. As shown in Fig. \ref{fig:height}, the building blocks the direct propagation path $O\rq X\rq$ only if its height $h>h_y$, where $h_y$ can be computed as
\begin{eqnarray}
h_y&=&\frac{yH_\mathrm{B}+(R-y)H_\mathrm{U}}{R}.\label{eqn3dh}
\end{eqnarray}
Next given that a building intersects the link $\bO\bX$, the intersection is uniformly distributed across the link, which indicates $y$ is uniformly distributed on $[0,\:R]$. Hence given a building intersects $\bO\bX$, the probability it blocks $\bO\rq \bX\rq$ is
\begin{align}
\eta&=\frac{1}{R}\int_{0}^{R}\mathbb{P}\left[h>h_y\right]\mathrm{d}y\nonumber \\
&=\frac{1}{R}\int_{0}^{R}\left(1-\int_{0}^{\frac{yH_\mathrm{B}+(R-y)H_\mathrm{U}}{R}}f_H(h)\mathrm{d}h\right)\mathrm{d}y\label{eqn3d1}\\
&=1-\int_{0}^{1}\int_{0}^{sH_\mathrm{B}+(1-s)H_\mathrm{U}}f_H(h)\mathrm{d}h\mathrm{d}s\label{eqn3d2},
\end{align}
where (\ref{eqn3d1}) is obtained by substituting (\ref{eqn3dh}), and (\ref{eqn3d2}) is from a change of variable as $s=\frac{y}{R}$.

Since $\eta$ is only determined by the distribution of the heights, which is independent of $K$, $\hat{K}$ can be viewed as the result of independent thinning of $K$ with a parameter of $\eta$. Hence $\hat{K}$ is also Poisson, and $\bbE[\hat{K}]=\eta\bbE[K]$.
\end{proof}

Note that incorporating the height of blockages only introduces a constant scaling factor $\eta$ to the results that ignore height. Hence, for simplicity, we will not consider the height of blockages in the following sections. The results can be readily modified to account for heights by incorporating the $\eta$ factor appropriately.
\subsection{Quantification of Power Losses by Blockages}
Now armed with the distribution of the number of blockages per link, we continue to derive the distribution of the penetration losses due to blockages on a link. The power losses caused by random blockages on a link is expressed as $S=\prod_{k=1}^{K}\gamma_k$, where $K$ is the number of blockages on the path, and $\gamma_k$ is the ratio of power losses due to $k$-th blockage. We have already established that $K$ is a Poisson variable with parameter $\beta R+p$.

If we use $[\:]$ to represent quantity in dB, then the ratio of penetration loss in dB is $[S]=\sum_{k=1}^{K}[\gamma_k]$, where $[\gamma_k]=10\log_{10}\gamma_k.$ Since we ignore reflections, $0<S\le 1$. Thus $[S]$ is always non-positive. It is useful to evaluate penetration loss in dB, for the statistics of penetration losses through different materials, i.e $\gamma_k$, are mostly recorded in dB.

Next, we will focus on deriving the distribution of $[S]$ with the assumption that $[\gamma_k]$ are i.i.d. with some distribution $f_{[\gamma_k]}(x)$. The Laplace transform of a random variable is used to simplify the computation. For a random variable $X$ in $\bbR^{+}$, its Laplace transform $\mathbb{L}_X(s)$ is defined as
\begin{eqnarray}
\mathbb{L}_X(s)=\bbE\left[\mathrm{e}^{-sX}\right]=\int_{0}^{\infty}f_X(x)\mathrm{e}^{-sx}\mathrm{d}x,
\end{eqnarray}
where $f_X(x)$ is the probability density function of $X$. Knowing the Laplace transform of $X$ is equivalent to knowing its distribution, for the probability density function can be obtained by computing the inverse Laplace transform of $\mathbb{L}_X(s)$. We obtain $\mathbb{L}_{[S]}(t)$ for general distributions of $[\gamma_k]$ in the following theorem.
\begin{theorem}
Suppose the length of the link is $R$, and $[\gamma_k]$ is i.i.d. variable with the Laplace transform $\mathbb{L}_{[\gamma_k]}(t)$. Then the Laplace transform of $[S]$ is
\begin{eqnarray}\label{thm:generalS}
\mathbb{L}_{[S]}(t)=\mathrm{e}^{\left(\beta R+p\right)\left(\mathbb{L}_{[\gamma_k]}(t)-1\right)}.
\end{eqnarray}
\end{theorem}
\begin{proof}
It follows that
\begin{align*}
\mathbb{L}_{[S]}(t)&=\bbE\left[\mathrm{e}^{-t\sum_{k=1}^{K}[\gamma_k]}\right]\\
&\stackrel{(a)}=\bbE_K\left[\left[\mathbb{L}_{[\gamma_k]}(t)\right]^K\right]\\
&\stackrel{(b)}=\mathrm{e}^{\left(\beta R+p\right)\left(\mathbb{L}_{[\gamma_k]}(t)-1\right)},
\end{align*}
where $\beta$ and $p$ are parameters of blockage processes defined in the previous part, $(a)$ is from the fact that $[\gamma_k]$ are i.i.d. variables, and $(b)$ follows that $K$ is a Poisson variable with parameter $\beta R+p$.
\end{proof}
For a generally distributed $[\gamma_k]$, the distribution of $[S]$ may be hard to compute from $L_[s](t)$ exactly, but can be obtained through numerical integrals. Closed form expressions can be found in certain special cases.
\begin{corollary}
In the impenetrable case where $\gamma_k=0$, $S$ is a Bernoulli random variable with parameter $\bbP\{S=1\}=\mathrm{e}^{-(\beta R+p)}$.
\end{corollary}
\begin{corollary}
Suppose that $\gamma_k$ is i.i.d. uniformly distributed on $[0,1]$. Then the probability density function of $[S]$ is
\begin{align}
f_{[S]}(x)=\mathrm{e}^{-(\beta R+p)}\delta(x)+\mathrm{e}^{-\left(\beta R+p-0.1\ln10x\right)}\sqrt{\frac{-0.1\ln10(\beta R+p)}{x}}I_1\left(2\sqrt{\frac{-\ln10}{10}x(\beta R+p)}\right),\label{thm:unifSdb}
\end{align}
where $x\le 0$, $\delta(x)$ is the Dirac delta function, and $I_1$ is the first-order modified Bessel function of the first kind.
The probability density function of $S$ is
\begin{eqnarray}
f_{S}(y)=\mathrm{e}^{-(\beta R+p)}\delta(y-1)+\mathrm{e}^{-(\beta R+p)}\sqrt{-\frac{\beta R+p}{\ln y}}\mathrm{I}_1\left(2\:\sqrt{-(\beta R+p)\ln y}\right),\label{thm:unifS}
\end{eqnarray}
where $y\in (0,1]$.
\end{corollary}
\begin{proof}
Note that the Laplace transform of $[\gamma_k]$ is
\begin{align}
\mathbb{L}_{[\gamma_k]}(t)=\frac{0.1\ln 10}{0.1\ln 10+t},\label{eqn:gamma}
\end{align}
and that the Laplace transform of $[S]$ can be expressed as
\begin{eqnarray}
\mathbb{L}_{[S]}(t)&=&\bbE_K\left[\left[\mathbb{L}_{[\gamma_j]}(t)\right]^K\right]\nonumber\\
&=&\sum_{n=0}^{\infty}\frac{\mathrm{e}^{-(\beta R+p)}(\beta R+p)^n}{n!}\left[\mathbb{L}_{[\gamma_j]}(t)\right]^n. \label{proof:unifS}
\end{eqnarray}
Substituting (\ref{eqn:gamma}) into (\ref{proof:unifS}), and compute the inverse Laplace transform of each term in the right hand side. This leads to
\begin{eqnarray*}
f_{[S]}(x)=\mathrm{e}^{-(\beta R+p)}\delta(x)+\sum_{m=0}^{\infty}\frac{\mathrm{e}^{-(\beta R+p-0.1\ln10x)}\left[0.1\ln10\left(\beta R+p\right)\right]^{m+1} (-x)^{m}}{(m)!\:(m+1)!},
\end{eqnarray*}
which can be simplified to (\ref{thm:unifSdb}) by using the modified Bessel function of the first kind. By changing variables as $y=\mathrm{e}^{0.1\ln10 x}$ in (\ref{thm:unifSdb}), we obtain (\ref{thm:unifS}).
\end{proof}

Even if $\gamma_k$ is a continuous random variable on $[0,1]$, $S$ always has a discontinuity at $x=1$ in the form of Dirac delta function $\delta(x-1)$, which indicates the probability of no power loss by blockage. Hence the amplitude of the impulse $\delta(x-1)$ equals the probability that a link admits line-of-sight propagation as computed in Corollary \ref{cor:free}.

Since the general closed-form expression for $f_S(x)$ is hard to obtain through the inverse Laplace transform, we consider an approximation of the continuous part of $f_S(x)$ using the beta distribution. The beta distribution has also been applied to model the behavior of random variables with support on intervals of finite length in a wide variety of disciplines \cite{Johnson1995}. We assume the target probability density function of the approximation has the form
\begin{eqnarray}
f_{\hat{S}}(x)=\left(1-\mathrm{e}^{-(\beta R+p)}\right)\frac{\Gamma(a+b)}{\Gamma(a)\Gamma(b)}x^{a-1}\left(1-x\right)^{b-1}+\mathrm{e}^{-(\beta R+p)}\delta\left(x-1\right), \label{eqn:approxS}
\end{eqnarray}
where $x\in (0,1]$, and $\Gamma(z)$ is the gamma function defined as $\Gamma(z)=\int_{0}^{\infty}t^{z-1}\mathrm{e}^{-t}\mathrm{d}t$.

The parameters $a$ and $b$ in (\ref{eqn:approxS}) can be determined by matching the moments of $S$. The moments of $S$ can be computed through the following theorem.
\begin{theorem}\label{thm:momentS}
Suppose that $\gamma_j$ are i.i.d. random variable on $[0,1]$. Then the $n$-th moment of $S$ is
\begin{eqnarray}
\bbE[S^n]=\mathrm{e}^{-(\beta R+p)(1-\bbE[\gamma_k^n])}.
\end{eqnarray}
\end{theorem}
\begin{proof}
The proof is straightforward as
\begin{align*}
\bbE[S^n]&
=\bbE_{K}\left[\left(\bbE[\gamma_k^n]\right)^K\right]\\
&\stackrel{(c)}=\mathrm{e}^{-(\beta R+p)(1-\bbE[\gamma_k^n])},
\end{align*}
where (c) follows from the fact that $K$ is a Poisson random variable with mean $\beta R+p$.
\end{proof}
The parameters in (\ref{eqn:approxS}) can be obtained by matching the moments of $S$ and $\hat{S}$.
\begin{corollary}
The parameters in the distribution of $\hat{S}$, i.e. in (\ref{eqn:approxS}), can be evaluated as
\begin{align*}
a=\frac{\left(\delta_2-\delta_1\right)\left(\delta_1-\delta_0\right)}{\left(\delta_1-\delta_0\right)^2-\left(\delta_2-\delta_0\right)\left(1-\delta_1\right)},
\end{align*}
\begin{align*}
b=\frac{\left(\delta_2-\delta_1\right)\left(1-\delta_1\right)}{\left(\delta_1-\delta_0\right)^2-\left(\delta_2-\delta_0\right)\left(1-\delta_1\right)},
\end{align*}
where $\delta_0=\mathrm{e}^{-(\beta R+p)}$, $\delta_1=\delta_0^{1-\mathbb{E}[\gamma_k]}$, and $\delta_2=\delta_0^{1-\mathbb{E}[\gamma_k^2]}$ are constants determined by the statistics of buildings in the area.
\end{corollary}
\begin{proof}
It follows directly from matching the first and second moment of $\hat{S}$ in (\ref{eqn:approxS}) and that of $S$ in Theorem \ref{thm:momentS}.
\end{proof}
Note that given $\beta$ and $p$, the second-order statistics of $\gamma_k$ is sufficient to determine the distribution of $\hat{S}$.

\subsection{Analysis of Blockage Effects on Link Budget}
In this section, we investigate how blockage effects the link budget by deriving a new path loss formulas that accounts for the penetration losses on a link. By Theorem \ref{thm:momentS}, we derive the following formula.
\begin{corollary}\label{cor:pathloss}
Considering the blockages in a network, the average received power $\bbE\left[P_\mathrm{r}\right]$ on a link of distance $R$ is
\begin{eqnarray}
\bbE\left[P_\mathrm{r}\right]=\frac{M\bbE[g] \bbE[S]}{R^{\alpha}}=\frac{M\: \:\mathrm{e}^{-(\beta R+p)(1-\bbE[\gamma_k])}}{R^{\alpha}},
\end{eqnarray}
where the parameters are the same as defined in (\ref{powerloss}).
\end{corollary}
Corollary \ref{cor:pathloss} indicates that blockage effects, on average, introduce an additional exponential decay in the link budget. This observation matches prior results in \cite{Franceschetti2004}, where data from field experiments was used to verify the statement. Note that Corollary \ref{cor:pathloss} provides a path loss formula for the general case, where the locations of the transmitter and receiver, whether indoor or outdoor, are not specified. We will provide the path loss formulas in both outdoor-to-outdoor and indoor-to-outdoor links in the following theorem.
\begin{theorem}\label{thm:outdoor}
The average received power for outdoor-to-outdoor links is
\begin{align}
\bbE\left[P_\mathrm{r}\right]=\frac{M\:\mathrm{e}^{-(\beta R-p)(1-\bbE[\gamma_k])}}{R^{\alpha}},\label{eqn:outout}
\end{align}
while for indoor-to-outdoor links, it is
\begin{align}
\bbE\left[P_\mathrm{r}\right]&=\frac{M\left(1-\mathrm{e}^{-\bbE\left[\gamma_k\right]p}\right)\mathrm{e}^{-\beta R(1-\bbE[\gamma_k])}}{R^{\alpha}\left(1-\mathrm{e}^{-p}\right)}\label{eqn:outin}\\
&\stackrel{(d)}\approx\frac{M\bbE[\gamma_k]\mathrm{e}^{-\beta R(1-\bbE[\gamma_k])}}{R^{\alpha}}.
\end{align}
\end{theorem}
\begin{proof}
If the link $\bO\bX$ is an outdoor-outdoor link, then both the transmitter $\bX$ and the receiver $\bO$ are not covered by blockages. By Theorem \ref{thm:general}, given that both $\bO$ and $\bX$ are not covered, the number of blockages is a Poisson variable with parameter $\bbE\left[V\left(\bO\bX\oplus D\right)-2V(D)\right]=\beta R-p$, where $D$ ranges over the object set of the Boolean scheme of blockages. Hence, given $\bO$ and $\bX$ are not covered, $\bbE\left[S\right]=\mathrm{e}^{-(\beta R-p)(1-\bbE[\gamma_k])}$, which leads to (\ref{eqn:outout}) directly.

Similarly, for an indoor-outdoor link, conditioning on that exactly one of $\bX$ and $\bO$ is covered by a blockage, $\bbE\left[S\right]=\frac{\left(1-\mathrm{e}^{-\bbE\left[\gamma_k\right]p}\right)\mathrm{e}^{-\beta R(1-\bbE[\gamma_k])}}{1-\mathrm{e}^{-p}}$, which leads to (\ref{eqn:outin}). Note that $p$ is the average ratio of land covered by buildings, which is smaller than 1. The approximation in $(d)$ follows from the fact that $\lim_{p\to0}\frac{1-\mathrm{e}^{-\bbE\left[\gamma_k\right]p}}{1-\mathrm{e}^{-p}}=\bbE\left[\gamma_k\right]$.
\end{proof}
\begin{remark}
By Theorem \ref{thm:outdoor}, an indoor cellular user will, on average, receive signals $\mathrm{e}^{-p\left(1-\bbE\left[\gamma_k\right]\right)}\bbE[\gamma_k]$ weaker than that received by an outdoor user. This motivates the wide deployment of indoor small cells like femtocells.
\end{remark}
\section{Analysis of Blockage effects on Impenetrable Networks} \label{apply}
In this section we analyze the effects of blockages on the system-level performance in cellular networks. For simplicity of the analysis, we focus on the case where the blockages are impenetrable.

To maintain the tractability of analysis, we make one key approximation: blockages affect each link independently. Strictly speaking, the number of blockages on different links are not always independent. For instance, if two base stations happen to locate on the same ray originating from the mobile user, then the base station further from the user will always have no fewer buildings on the link than the closer one has. Thus these two links are correlated. There are cases, however, where the number of the blockages on the two links are independent. For example, whenever two links share no common blockages, the numbers of blockages on the links are independent. Though an approximation, simulations show that ignoring the correlation of shadowing between links causes minor loss in accuracy, especially when the sizes of blockages $L_i$ and $W_i$ are relatively small compared with the lengths of links $\bO\bX_i$. Moreover, under the independent link assumption, simple expressions of performance metrics, such as network connectivity and coverage probability, are derived for the references of system design and analysis. Discussion on the correlation of shadowing among links can be found in \cite{Oestges2011,Patwari2008}

\subsection{Connectivity}
In a cellular network with impenetrable blockages, the blockages divide the plane into many isolated ``islands". Only the locations within the same island can communicate directly via wireless links. In this section, we investigate the connectivity of an impenetrable networks. Specifically, we aim to answer the following questions: (\romannumeral 1) How many candidate base stations, i.e. the unblocked base stations, are there for a typical user? (\romannumeral 2) What is the distribution of the distance to the nearest candidate base station?

To analyze the connectivity, we make some additional definitions. For $\bX,\bY\in\bbR^2$, we define $K_{XY}$ as the number of blockages on the direct link connecting $\bX$ and $\bY$. Moreover, $\bX$ is {\it visible} by $\bY$ if and only if $K_{XY}=0$, i.e. no blockages intersects the link $\bX\bY$. Intuitively, the {\it visible area} of an outdoor location $\bX$ is the set of locations which can be connected to $\bX$ with a line-of-sight link. Denote the visible region of $\bX$ as $Q_X$. A formal definition of $Q_X$ is introduced as follows.
\begin{definition}
Suppose that $\bX \in \mathbb{R}^2$ is not covered by a blockage. The visible region of $\bX$ is
\begin{align*}
Q_X=\{\bY\in\bbR^2:K_{XY}=0\}.
\end{align*}
If $\bX$ is covered by a blockage, then $Q_X = \emptyset$, i.e, it has no visible area.
\end{definition}
Note that in an impenetrable network, a mobile user at $\bX$ can only connect to base stations located in its visible area $Q_X$. The average size of visible area can be computed through the following theorem.
\begin{theorem}\label{thm:average}
The average size of a visible region in a cellular network with impenetrable blockages is
$
\bbE[V(Q_x)]=\frac{2\pi\mathrm{e}^{-p}}{\beta^2}.
$
\end{theorem}
\proof
Since the blockages are modeled as a Boolean scheme of rectangles, which is stationary in $\bbR^2$, it is sufficient to check the visible area of the origin $Q_0$. Denote $A$ as the event that the origin is not covered by a blockage. By Corollary \ref{cor:pointcover},
$
\mathbb{P}\{A\}=\mathrm{e}^{-p}.
$
If $A$ is not true, then $Q_0=0$. Thus we focus on the case when the origin is not coverd by any blockage.

Define $\{\mathbf{a}_{\psi}\}$ as the set of normalized vectors, where $\psi$ is the angle between $\mathbf{u}_{\psi}$ and $x$ axis. Define the distance to the nearest blockage along the direction $\mathbf{u}_{\psi}$ as
\begin{eqnarray}
D(\psi) = \sup\{r\in\bbR^{+}:K_{r\mathbf{u}_{\psi}}=0\}.
\end{eqnarray}

Next, by Corollary \ref{cor:free},
\begin{align*}
\bbP\{D({\psi})>R|A\}&=\frac{\bbP\{K_{R\mathbf{u}_{\psi}}=0\}}{\bbP\{A\}}\\
&=\frac{\mathrm{e}^{-(\beta R+p)}}{\mathrm{e}^{-p}}=\mathrm{e}^{-\beta R}.
\end{align*}
Namely, conditioning on $A$, $D(\psi)$ is an exponential random variable with parameter $\beta$. Moreover, the distribution of $D(\psi)$ is independent of $\psi$.
Hence given $A$, the average size of the visible area is
\begin{align*}
\bbE[V(Q_0)|A]&=\mathbb{E}\left[\int_{0}^{2\pi}\frac{D^2(\psi)}{2}\mathrm{d}\psi\right]\\
&=2\pi\int_{0}^{\infty}\frac{r^2}{2}\beta\mathrm{e}^{-\beta r}\mathrm{d}r=\frac{2\pi}{\beta}.
\end{align*}
Finally, the average size of the visible area is
\begin{align*}
\bbE[V(Q_0)]&=\bbP\{A\}\bbE[V(Q_0)|A]+(1-\bbP\{A\})\times0\\
&=\frac{2\pi\mathrm{e}^{-p}}{\beta^2}.  \tag*{\mbox{\endproof}}
\end{align*}
We define the {\it effective visible range} of a network  $R_\mathrm{eff}$ as the radius of a circle whose area equals the average visible area $\bbE[V(Q_0)]$ of the network.
\begin{corollary}
In a cellular network with impenetrable blockages, $R_\mathrm{eff}=\frac{\sqrt{2\mathrm{e}^{-p}}}{\beta}$.
\end{corollary}
\begin{proof} By Theorem \ref{thm:average}, $R_\mathrm{eff}=\sqrt{\frac{\bbE[V(Q_0)]}{\pi}}=\frac{\sqrt{2\mathrm{e}^{-p}}}{\beta}.$
\end{proof}
Note that $R_\mathrm{eff}$ can reveal the average range that a base station can reach via line-of-sight links in a network.

The average number of base stations visible to a mobile user is derived in the following corollary.
\begin{corollary}
In cellular networks with impenetrable blockages, if the base stations form a homogeneous PPP with density $\mu$, then the average number of visible base stations to a mobile user is $\mu\bbE[V(Q_0)]=\frac{2\pi\mu\mathrm{e}^{-p}}{\beta^2}$.
\end{corollary}
\begin{proof}
Note that $V(Q_0)$, the size of the visible region of the mobile user, is a random variable. We denote its probability density function as $f_V(v)$.
By Assumption \ref{assump:PPP}, the PPP of base stations is independent of the Boolean scheme of blockages. Hence, given that the size of the visible area is $v$, there are $\mu v$ visible base stations on average. Averaging over all the realizations of blockages, the average number of visible base stations is
$$ \int_{0}^{\infty}\mu v f_V(v)\mathrm{d}v=\mu\bbE[V(Q_0)].$$
Lastly, by Theorem \ref{thm:average}, it follows that $\mu\bbE[V(Q_0)]=\frac{2\pi\mu\mathrm{e}^{-p}}{\beta^2}$.
\end{proof}
Note that the average number of base stations visible to a mobile user is finite. This indicates that there are only a finite number of base stations visible to a mobile user almost surely in a cellular network with impenetrable blockages.

Next we derive the distribution of the distance to the closest visible base stations under the assumption that the number of blockages on each link is independent.
\begin{theorem}\label{thm:Rfree} Assuming that the numbers of blockages on the links are independent, then the distribution of the distance to the closest visible base station $R_0$ is
\begin{eqnarray}\label{eqn:cdfR}
\bbP\{R_{0}> x\} = \exp(-2\pi\mu U(x)),
\end{eqnarray}
 where $U(x)= \frac{\mathrm{e}^{-p}}{\beta^2}\left[1-(\beta x+1)\mathrm{e}^{-\beta x}\right]$.
\end{theorem}
\proof
Suppose that the mobile user is located at the origin. The distance to the nearest visible BS $R_0$ is larger than $x$ if and only if all the base stations located within the ball $\mathcal{B}(\bO,x)$ are not visible to the user. Since the base stations form a PPP of density $\mu$, it follows that
\begin{align*}
\bbP\{R_{0}> x\}&=\mathbb{P}\{\mbox{all base stations in } \mathcal{B}(\bO,x) \mbox { are not visible} \}\\
&=\sum_{i=0}^{\infty}\left[\int^x_0\left(1-\mathrm{e}^{-(\beta t+p)}\right)\frac{2t}{x^2}\mathrm{d}t\right]^i\frac{\mathrm{e}^{-\mu \pi x^2}(\mu \pi x^2)^i}{i!}\\
&=\sum_{i=0}^{\infty}\left(1-\frac{2U(x)}{x^2}\right)^i\frac{\mathrm{e}^{-\mu \pi x^2}(\mu \pi x^2)^i}{i!}\\
&=\mathrm{e}^{-\mu\pi x^2\left[1-\left(1-\frac{2U\left(x\right)}{x^2}\right)\right]}\\
&=\mathrm{e}^{-2\mu\pi U(x)}. \tag*{\mbox{\endproof}}
\end{align*}
One direct corollary of Theorem \ref{thm:Rfree} follows by differentiating  (\ref{eqn:cdfR}).
\begin{corollary}\label{rfreepdf}
The probability density function of $R_0$ is
$f_{R_0}(x) = 2\pi\mu x\mathrm{e}^{-(\beta x+p+2\pi\mu U(x))}$.
\end{corollary}

In a network with impenetrable blockages, a mobile user can only be  served by the base stations in his visible area. If there happens to be no base stations located inside its visible region, the mobile user will receive no signals. We define a {\it silent} area of an impenetrable network as the set of all locations that have no visible base stations. Denote $\xi$ as the ratio of the silent area to the total area of a network. We estimate the value of $\xi$ in the following corollary.
\begin{corollary} \label{cor:silent}
Under the independent link assumption, the ratio of the silent area to the total area in a network with impenetrable blockages is
\begin{eqnarray}
\xi=1-\bbP(R_0<\infty) = \exp\left(-\frac{2\mu\pi\mathrm{e}^{-p}}{\beta^2}\right).\label{eqn:silentratio}
\end{eqnarray}
\end{corollary}
\proof
Since the base stations and blockage processes are stationary on the plane, the ratio of the silent area to the total area of a network equals the probability that a user at the origin has no visible base stations, which also equals the probability that the closest visible base station to the origin is infinitely far away. It follows that
\begin{align*}
\xi&=1-\bbP\{R_0<\infty\}=\lim_{x\rightarrow \infty} \bbP\{R_0>x\}\\
&= \exp\left[-2\mu\pi\mathrm{e}^{-p}\lim_{x\rightarrow \infty}\left(\frac{1}{\beta^2}\left(1-\mathrm{e}^{-\beta x}(1+x)\right)\right)\right]\\
&= \exp \left(\frac{-2\mu\pi\mathrm{e}^{-p}}{\beta^2}\right). \tag*{\mbox{\endproof}}
\end{align*}
\begin{remark}
The parameter $\xi$ illustrates the level of connectivity in a network. When $\xi$ increases, users become less likely to get connected in the network. Substituting $\beta=\frac{2\lambda(\bbE[L]+\bbE[W])}{\pi}$ for (\ref{eqn:silentratio}), it follows that
\begin{eqnarray*}
\xi=1- \exp\left(\frac{-\pi^3\mu}{2\lambda^2\mathrm{e}^{-\lambda\bbE[W]\bbE[L]}\left(\bbE[L]+\bbE[W])\right)^2}\right).
\end{eqnarray*}
This indicates that to keep the connectivity fixed, the density $\mu$ of base stations should scale superlinearly with the blockage density $\lambda$. Specifically when $p$ is small, $\mu$ should scale with $\lambda^2$ approximately.
\end{remark}

\subsection{Coverage Probability}
The coverage probability is an important performance metric in a cellular network. In an interference-limited network, the coverage probability is defined as
\begin{eqnarray}
P_c(T)=\bbP\{\mbox{SIR}>T\},
\end{eqnarray}
where $T>0$ is the SIR threshold for successful decoding at the receiver. If the base station process is stationary on the plane, the coverage probability can be also interpreted as the percentage of the area in a network where the received SIR is higher than $T$.

While not expressed explicitly, $P_c(T)$ is a function of base station density $\mu$, blockage density $\lambda$, and other statistics of blockages, such as $\bbE[L]$ and $\bbE[W]$. More importantly, $P_c$ also depends on the connecting strategy of the user. For instance, a mobile user can connect to the base station that provides maximum SIR or the strongest signal power. For tractability in the analysis, we assume the mobile user always connects to the {\it nearest} visible base station (if one is visible). We also assume that each link has an independent number of blockages. In this case, the coverage probability can be computed through the following theorem.
\begin{theorem}\label{thm:cover}
If the user connects to the nearest visible base station, the coverage probability $P_c(T)$ is
\begin{align}
P_{c}(T)=\int_{0^{+}}^{\infty}\exp\left(-2\pi\mu\int_{x}^{\infty}\left[\frac{Tx^\alpha\mathrm{e}^{-(\beta t+p)}}{t^{\alpha}+Tx^\alpha}\right]t\mathrm{d}t\right)f_{R_0}(x)\mathrm{d}x,
\end{align}
where $f_{R_0}(x)$ is the probability density function of the distance to the nearest visible base station derived in Corollary \ref{rfreepdf}, $\alpha$ is the path loss exponent as in (\ref{powerloss}).
\end{theorem}
\proof
First we compute the expression of coverage probability conditioning on the distance to the nearest visible base station $R_0$. Given that $R_0=x$, the expression of SIR is
\begin{eqnarray*}
\mbox{SIR}=\frac{x^{-\alpha}g_0}{\sum_{i:R_i>x}R_i^{-\alpha}g_i S_i},
\end{eqnarray*}
where $g_i$ are i.i.d. exponential random variables with mean 1, $S_i$ (in the impenetrable case) are independent Bernoulli random variables with parameter $\mathrm{e}^{-(\beta R_i+p)}$. The conditional coverage probability follows as
\begin{align*}
\bbP\{\mbox{SIR}>T|R_0=x\}&=\bbP\{g_0>Tx^\alpha\sum_{i:R_i>x}R_i^{-\alpha}g_i S_i\}\\
&=\bbE\left[\exp\left(-Tx^{\alpha}\sum_{i:R_i>x}R_i^{-\alpha}g_i S_i\right)\right]\\
&=\bbE\left[\prod_{i:R_i>x}\bbE_{S_i,g_i}\left[\exp\left(-Tx^{\alpha}R_i^{-\alpha}g_i S_i\right)\right]\right]\\
&=\bbE\left(\prod_{i:R_i>x}\bbE_{g_{i}}\left[\exp(-Tx^{\alpha}g_{i}R_{i}^{-\alpha})\right]\mathrm{e}^{-(\beta R_{i}+p)}+1-\mathrm{e}^{-(\beta R_i+p)}\right)\\
&=\bbE\left(\prod_{i:R_i>x}1-\frac{Tx^{\alpha}\mathrm{e}^{-(\beta R+p)}}{R_i^{\alpha}+Tx^{\alpha}}\right)\\
&\stackrel{(e)}{=}\exp\left(-2\pi\mu\int_{x}^{\infty}\left[\frac{Tx^\alpha \mathrm{e}^{-(\beta t+p)}}{t^{\alpha}+Tx^\alpha}\right]t\mathrm{d}t\right),
\end{align*}
where $(e)$ follows directly from computing Laplace functional of the PPP formed by the base stations $\{X_i\}$ \cite{Baccelli2009}. The unconditional coverage probability follows as
\begin{align*}
P_c(T)
&=\int_{x>0}\bbP(\mbox{SIR}>T|R_0=x)f_{R_0}(x)\mathrm{d}x\\
&=\int_{0^{+}}^{\infty}\exp\left(-2\pi\mu\int_{x}^{\infty}\left[\frac{Tx^\alpha\mathrm{e}^{-(\beta t+p)}}{t^{\alpha}+Tx^\alpha}\right]t\mathrm{d}t\right)f_{R_0}(x)\mathrm{d}x.\tag*{\mbox{\endproof}}
\end{align*}
\begin{remark}
By Theorem \ref{thm:cover}, the coverage probability is not invariant with base station density in a interference-limited network with blockages, while it was shown that, without considering blockage effects, the coverage probability is independent of the base station density in \cite{Andrews2011}. In brief, in a network without blockages, the strengths of the desired signal and interference scale by a common factor when altering the base station density. In a network with blockages, however, the scale factors are different for signal and interference power, for interference links are expected to have more blockages on average.
\end{remark}
\subsection{Average Achievable Rate}
The average achievable rate is another important performance metric in a wireless system. It gives an upper bound of the average data rate that a cellular network can support. The average achievable rate $\tau$ is defined as
\begin{eqnarray}
\tau=\bbE[\log_2(1+\mbox{SIR})].
\end{eqnarray}

By \cite{Andrews2011}, given the coverage probability $P_c(T)$ in a cellular network with impenetrable blockages, the average achievable rate $\tau$ is
\begin{eqnarray}\label{thm:capacity}
\tau=\frac{1}{\ln 2}\int_{T>0}\frac{P_c(T)}{T+1}\mathrm{d}{T}.
\end{eqnarray}
It is possible that there is exactly one base station located in the visible area of a mobile user, though the possibility is small. In this case, the mobile user will not receive interference, which renders the SIR infinite. Hence, the convergence of (\ref{thm:capacity}) is very slow if it converges. In reality, even in absence of the thermal noise and interference, the maximum achievable SINR is limited by the distortion introduced by the RF imperfection, which is measured by the error vector magnitude (EVM) \cite{sesia2009}. For instance, the highest EVM requirement for transmitters in a LTE system is 0.08 \cite{3GPPTechnical}, which approximately sets the maximum possible SINR to be $10\log_{10}\frac{1}{\mathrm{EVM}^2}=50$ dB \cite{Gharaibeh2004}. Hence, for the purpose of rate calculation, it is reasonable to assume that the operating SIR is upper bounded by a threshold $T_{\max}$:
\begin{eqnarray}\label{SIR}
\mbox{SIR}=\min\left\{\frac{x^{-\alpha}g_0}{\sum_{i:R_i>x}R_i^{-\alpha}g_i S_i},T_{\max}\right\}.
 \end{eqnarray}
With the refined SIR expression in (\ref{SIR}), the average rate can be evaluated as
\begin{eqnarray}
\tau=\frac{1}{\ln2}\int_{0}^{T_{\max}}\frac{P_c(t)}{t+1}\mathrm{d}t.
\end{eqnarray}
\section{Simulation Results} \label{sec:simu}
In this section, we present numerical results that compare our proposed model with prior work, and illustrate how blockage effects impact the network performance.
\begin{figure} [!ht]
\centerline{
\includegraphics[width=0.6\columnwidth]{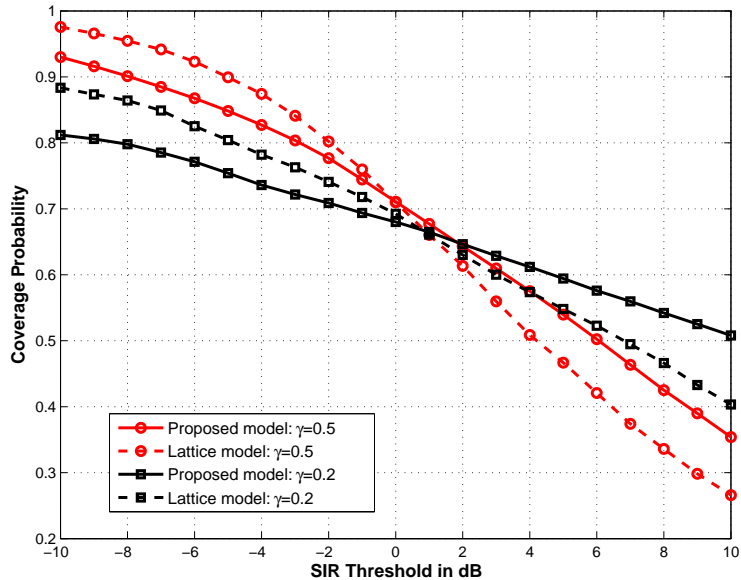}
   }
\caption{Comparison between the proposed Boolean scheme model and the lattice model. All curves are drawn according to Monte Carlo simulations, where $p=0.3$, $\bbE[L]=\bbE[W]=15$ meters. The ratio of penetration loss per building $\gamma$ is assumed to be constant. In the Boolean scheme model, we assume the lengths and widths of the blockages are i.i.d. uniformly distributed. Difference in the SIR distributions is observed when using different blockage models, which indicates that the distributions of the blockage orientation and size affect the network performance.}\label{fig:comparison}
\end{figure}

\begin{figure} [!ht]
\centerline{
\includegraphics[width=0.9\columnwidth]{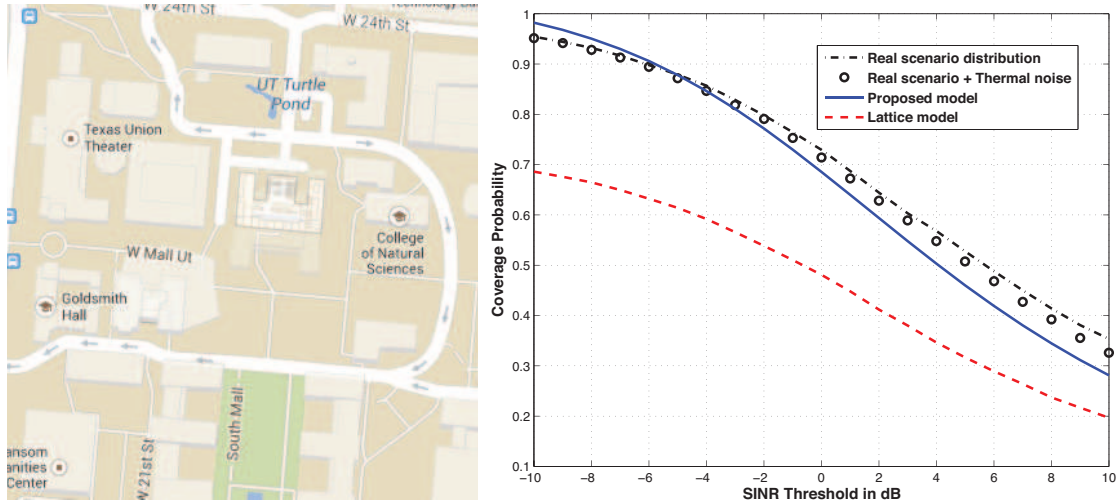}
   }
\caption{Comparison between the proposed blockage model and the real building distribution on the campus of The University of Texas at Austin. The snapshot of the campus is taken from Google map. We simulate a square area of $400\times 400$ $\mbox{m}^2$. We assume the buildings are impenetrable. In all simulations, the base stations are distributed according to a PPP. The average cell radius is 50 m. The network is assumed to be operated at $3.5$ GHz with a system bandwidth of $50$ MHz. The transmit power $P_\mathrm{t}$ is $30$ dBm. The thermal noise is not considered in the simulations unless specified. In the Boolean scheme model, we assume that $\mathbb{E}[L]=55$ m, $\mathbb{E}W=52$ m, and $p=0.266$, which match the real statistics of the area. To make a fair comparison, the rectangle site is of  the size $\bbE[L]\times\bbE[W]$, and occupied by a building with probability $p$ in the lattice model. }\label{fig:real_simu}
\end{figure}

\begin{figure} [!ht]
\centerline{
\includegraphics[width=0.6\columnwidth]{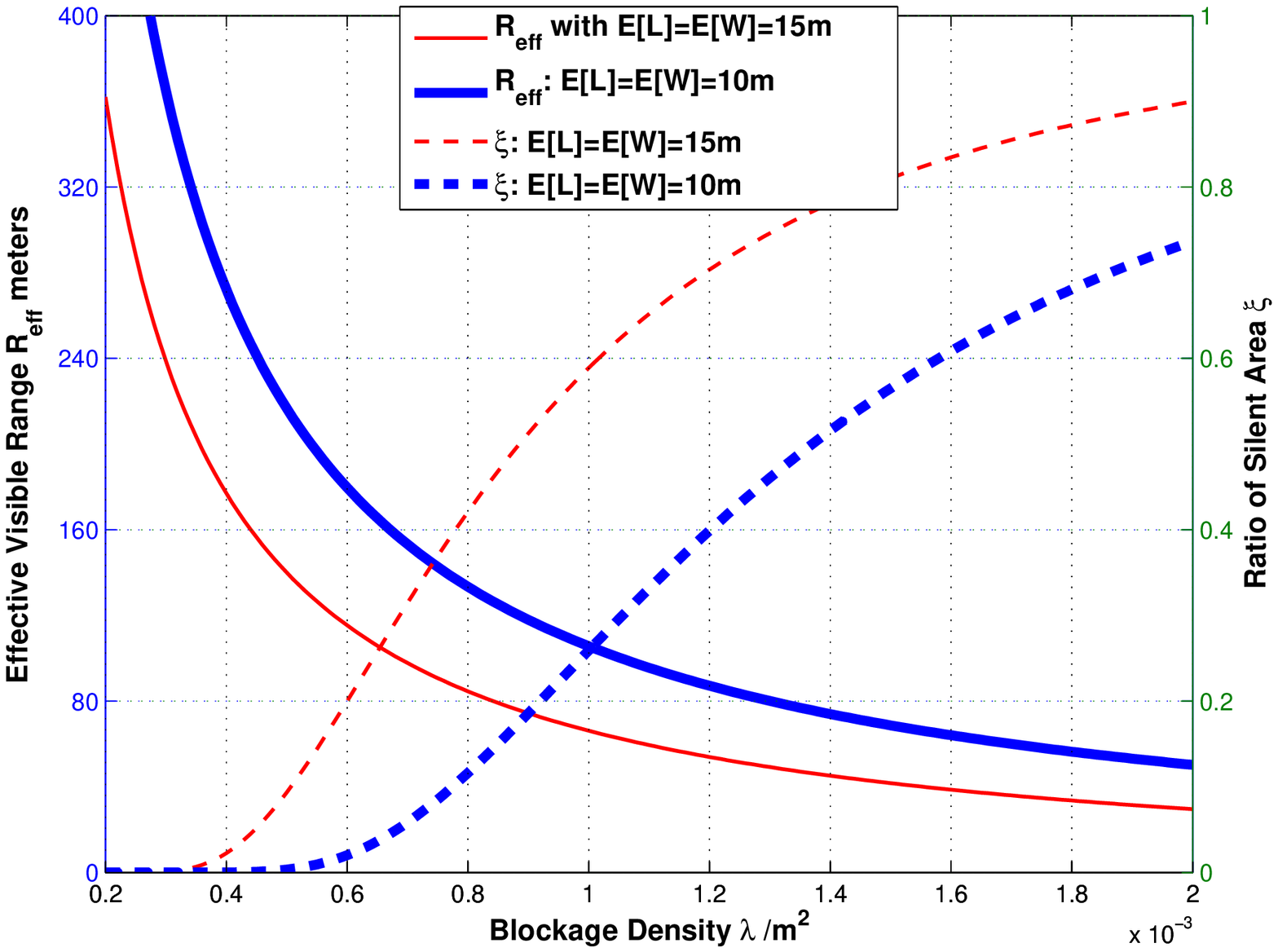}
   }
\caption{The network connectivity as a function of blockage density, in terms of the average visible area and the fraction of the silent area. The density of base station is $\mu_0=3.85\times10^{-5}$, where the average cell radius is 100 meters. When the density and the size of blockages increase, the visible area shrinks as the size of the silent area increases.}\label{fig:visible}
\end{figure}

First, we compare the simulation results based on the proposed Boolean scheme model and the lattice model in \cite{Franceschetti1999}. The original lattice model proposed in \cite{Franceschetti1999} assumes a random lattice of unit squares, each of which is occupied by a blockage with some probability. To make a fair comparison, we extend the model to a lattice made up of rectangle sites of size $\bbE[L]\times\bbE[W]$. Each site is occupied by a blockage with probability $p$, where $\bbE[L]$, $\bbE[W]$ and $p$ are parameters in the proposed Boolean scheme model, such that the average number of blockages on a link is the same in both the Boolean scheme model and the lattice model.

In \figref{fig:comparison}, we compare the simulated SIR distribution for both models with equivalent parameters. Though the average number of blockages in each link is identical, the SIR distributions are different due to the fact that the Boolean scheme model allows the blockage orientation and size to be random. Hence we conclude that the randomness of size and orientation is a differentiating feature of our model.

In \figref{fig:real_simu}, we compare the analytical blockage models with the building distribution in a real scenario. We take a 400 $\times$ 400 $\mathrm{m}^2$ snapshot of The University of Texas at Austin campus from Google map. In this area, the buildings take up 26.6$\%$ of the total size, and their average size is $55\times 52$ $\mbox{m}^2$. The parameters of the proposed Boolean scheme model and the lattice model are taken to match the statistics of the real buildings. We assume the base stations form a PPP with an average cell radius of $50$ m. We simulate the coverage probability of an outdoor user located at the center of the area. It is shown in \figref{fig:real_simu} that ignoring thermal noise causes minor errors in the simulation, for the dense deployment of base stations renders the network interference-limited. Simulation results also show that the Boolean scheme model may fit the real scenarios better than the lattice model, for it can incorporate the random size and location of the buildings.

Next, we consider how impenetrable blockages impact cellular networks. Simulations of the network connectivity, in terms of the average visible area and the fraction of the silent area in a network are illustrated in Fig. \ref{fig:visible} as a function of the density and size of the blockages. Results in Fig. \ref{fig:visible} indicate that the existence of blockages will limit the range of line-of-sight links, and increase the size of the silent area in a network. Though the existence of the blockages decreases the network connectivity, however, it may be helpful to increase the coverage probability and achievable rate as shown in the following simulations.

\begin{figure} [!ht]
\centerline{
\includegraphics[width=0.6\columnwidth]{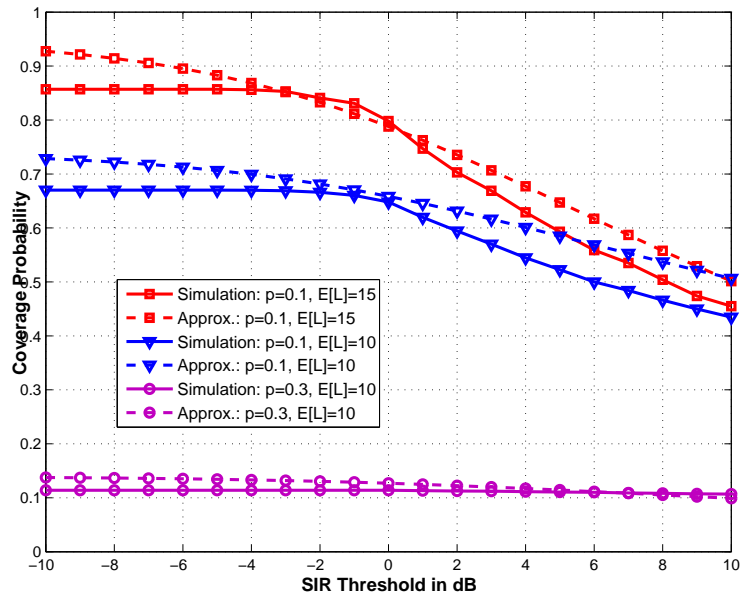}
   }
\caption{Comparison between analytical and numerical results via Monte Carlo simulation. We assume $W_i$ and $L_i$ are i.i.d. uniformly distributed in the simulations. The error is small in the practically relevant SIR regime between -5dB and 5dB.}\label{fig:coveragesim}
\end{figure}
We present a comparison of the coverage probability between Monte Carlo simulations and the analytical results from Section \ref{apply} in Fig. \ref{fig:coveragesim}. The minor difference between the curves are caused by neglecting correlations of blockage effects on different links in the analysis. The small error seems to be a reasonable compromise between fidelity and simulation time, since it takes more than 3 hours to run a Monte Carlo simulation of 10000 samples, and less than 1 minute to evaluate the coverage using the analytical expression.

\begin{figure} [!ht]
\centerline{
\includegraphics[width=0.6\columnwidth]{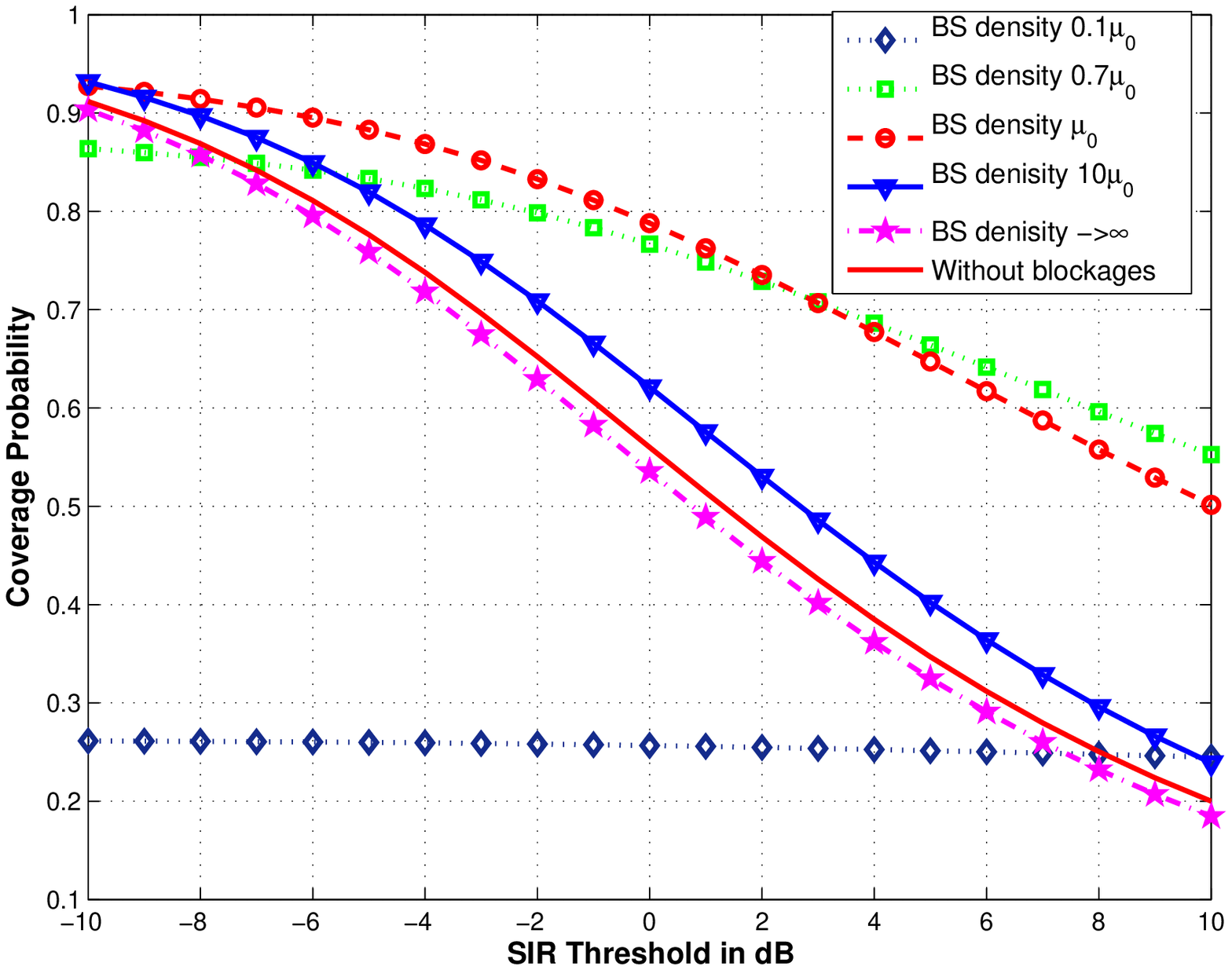}
   }
\caption{Coverage probability as a function of the base station density. Coverage probabilities are computed through the analytical expression in Theorem \ref{thm:cover}, where $\mu_0=3.85\times10^{-5}/\mbox{m}^{2}.$ The baseline curve with no blockages is derived according to \cite{Andrews2011}. When the base station density is $\mu_0$, a remarkable performance gain in coverage probability is observed over the baseline curve, which indicates blockages sometimes help to improve performance of urban networks.}\label{fig:coverage}
\end{figure}
The coverage probability is compared as a function of base station density, with and without blockages in \figref{fig:coverage}. When blockages are not considered in the system model, prior work in \cite{Andrews2011, Dhillon2012} concluded that the coverage probability is invariant with the base station density. When blockages are included, however, the SIR distribution becomes dependent on the base station density, as shown in \figref{fig:coverage}. An interesting result is that blockages often help improve coverage. This apparently counter-intuitive result can be explained by the distance-dependent behavior of blockage. Since we assume the user always connects to the nearest available base station, the interference link is always longer than the desired-signal link. Thus it is likely that there are more blockages on the interference link. Then a large portion of interference power will be blocked from the mobile user than that of the signal power. Another interesting result is the coverage probability as a function of different base station densities $\mu$ and 10$\mu$: increasing the number of base stations deployed in a region, sometimes leads to a degradation of performance. This simulation result indicates that, with a fixed blockage density, the optimum base station density to achieve best network performance is finite. Another observation is that when the base station density goes to infinity, the coverage probability converges to the case with no blockages. Since the locations of the base stations and buildings form two independent homogeneous PPPs, increasing base station density to infinity with a fixed blockage density is equivalent to decreasing the blockages density to zero with some base station density.
\begin{table}
 \begin{center}
\caption{Average Rate Comparison With $T_{\max}=40dB$}\label{table:rate}
\begin{tabular}{c|c|c|c|c }
Blockage density & None & Low & Intermediate & High \\ \hline
\hline
    \mbox{$\lambda/ \mu$} & 0 & $0.1 \lambda_0/\mu_0$ & $\lambda_0/\mu_0$ &  $10\lambda_0/\mu_0$\\ \hline 
    \mbox{Average rate (bits/sec/Hz)} &2.15 & 2.42 & 4.99& 3.14\\ \hline
    \end{tabular}
    \end{center}
    \footnotesize{{\bf Note:} $\lambda/\mu$ is the ratio of the blockage density to the base station density. We assume $\bbE[L]=\bbE[W]=15 m$, $\lambda_0=4.4\times 10^{-4}/\mbox{m}^{2}$, and $\mu_0=3.85\times 10^{-5}/\mbox{m}^{2}$. }
  \end{table}
  
Simulation results of the average rate provide similar insight that, considering blockage effects, the average rate is no longer invariant with base station density, and that blockages may help increase the achievable rate, as shown in Table \ref{table:rate}.
\section{Conclusions} \label{conclu}
In this paper, we proposed a stochastic framework to model random blockages in urban cellular networks using  concepts from random shape theory. The key idea is to model the random buildings as a Boolean scheme of rectangles, which allows for a comprehensive characterization of the randomness of the blockages, such as sizes, locations, orientations, and heights. Based on the blockage model, we derived the distribution of the penetration loss in a link, and proposed path loss formulas that incorporate the blockage effect in different scenarios. The proposed model captured the distance-dependent feature of the blockage effects. Analysis of the performance in cellular networks with impenetrable blockages indicates that blockages change the behavior of cellular networks in an important way, such as reducing the network connectivity and removing the invariance with base station density in terms of SIR distribution. Our results illustrate that cellular networks may benefit from blockages, since the longer path lengths to interfering base stations may have more blockages. For future work, it would be interesting to extend these results to multi-tier networks. It would also be interesting to consider the case where the infrastructure is deployed with some relationship to the blockages, for example small cells could be deployed inside blockages or base stations could be deployed on the perimeters of blockages. Further refinements of the proposed model would be useful, especially to consider reflections which are an important contributor to coverage in urban areas.

\end{document}